\documentclass{llncs}
\usepackage{amsmath}
\usepackage{amssymb}
\usepackage{algorithmicx}
\usepackage{algpseudocode}
\usepackage{graphicx}

\spnewtheorem*{corollary*}{Corollary}{\bf}{\itshape}

\algtext*{EndFunction}% Remove "end function" text
\algtext*{EndProcedure}% Remove "end procedure" text
\algtext*{EndFor}% Remove "end for" text
\algtext*{EndWhile}% Remove "end while" text
\algtext*{EndIf}% Remove "end if" text
\algloopdefx{NIf}[1]{\textbf{if} #1 \textbf{then}}
\algloopdefx{NElse}{\textbf{else}}
\algloopdefx{NElseIf}{\textbf{else if}}
\algloopdefx{NForAll}[1]{\textbf{for each} #1 \textbf{do}}
\algloopdefx{NWhile}[1]{\textbf{while} #1 \textbf{do}}
\algloopdefx{NFor}[1]{\textbf{for} #1 \textbf{do}}

\title {Online Square Detection}

\author
{
	Dmitry Kosolobov
}

\institute{Ural Federal University}

\begin{document}

\maketitle

\begin{abstract}
The online square detection problem is to detect the first occurrence of a square in a string whose characters are provided as input one at a time. Recall that a square is a string that is a concatenation of two identical strings. In this paper we present an algorithm solving this problem in $O(n\log\sigma)$ time and linear space on ordered alphabet, where $\sigma$ is the number of different letters in the input string. Our solution is relatively simple and does not require much memory unlike the previously known online algorithm with the same working time. Also we present an algorithm working in $O(n\log n)$ time and linear space on unordered alphabet, though this solution does not outperform the previously known result with the same time bound.
\end{abstract}

\section{Introduction}

The study of algorithms for analysis of different kinds of periodicities in strings constitutes an important branch in stringology and squares often play a central role in such researches. Recall that a string $s$ is a \emph{square} if $s = xx$ for some nonempty string $x$. A string is \emph{squarefree} if it does not have a substring that is a square. We consider algorithms recognizing squarefree strings. To be more precise, let $f$ be a positive function of integer domain; we say that an algorithm \emph{detects} squares in $O(f(n))$ time if for any integer $n$ and any string of length $n$, the algorithm decides whether the string is squarefree in $O(f(n))$ operations. We say that an algorithm detects squares \emph{online} if the algorithm processes the input string sequentially from left to right and decides whether each prefix is squarefree after reading the rightmost letter of that prefix.

In this paper we give two algorithms for online square detection. The first one works on ordered alphabet and requires $O(n\log\sigma)$ time and linear space, where $\sigma$ is the number of different letters in the input string. Though the proposed result is not new (see~\cite{HongChen}), it is rather less complicated than the previously known solution and can be used in practice. The second algorithm works on unordered alphabet and takes $O(n\log n)$ time and linear space. This algorithm is substantially different from the algorithm of~\cite{ApostolicoBreslauer} having with the same time and space bound. These two algorithms are equally applicable for practical use. Let us point out some previous results on the problem of square detection.

One can easily show that on a two-letter alphabet any string of length at least four is not squarefree. A classical result of Thue \cite{Thue} states that on a three-letter alphabet there are infinitely many squarefree strings. Main and Lorentz \cite{MainLorentz} presented an algorithm that detects squares in $O(n\log n)$ time and linear space for unordered alphabet. They proved an $\Omega(n\log n)$ lower bound for the problem of square detection on unordered alphabets, so their result is the best possible in this case. For ordered alphabets, Crochemore \cite{Crochemore} described an algorithm that detects squares in $O(n\log \sigma)$ time and linear space.

The interest in algorithms for online square detection was initially motivated by problems in the artificial intelligence research (see~\cite{LeungPengTing}). Leung, Peng, and Ting \cite{LeungPengTing} obtained an online algorithm that detects squares in $O(n\log^2 n)$ time and linear space on unordered alphabet. Jansson and Peng \cite{JanssonPeng} found an online algorithm that detects squares in $O(n(\log n + \sigma))$ time; for ordered alphabets, their algorithm requires $O(n\log n)$ time. Hong and Chen \cite{HongChen} presented an online algorithm that detects squares in $O(n\log\sigma)$ time and linear space on ordered alphabet. Their algorithm heavily relies on the amount of space consumed by string indexing structures and hence is rather impractical; it seems that even the most careful implementations of the algorithm use at least $40n$ bytes in the worst case. Apostolico and Breslauer \cite{ApostolicoBreslauer} as a byproduct of their parallel algorithm for square detection obtained an online algorithm that detects squares in $O(n\log n)$ time and linear space on unordered alphabet (apparently, the authors of~\cite{LeungPengTing} and~\cite{JanssonPeng} did not know about this result).

The present paper is inspired by Shur's work \cite{Shur} on random generation of square-free strings.

The paper is organized as follows. In Section~\ref{SectCatcher} we present some basic definitions and the key data structure called catcher, which helps to detect squares. Section~\ref{SectUnord} contains an online solution for the case of unordered alphabet. In Section~\ref{SectOrd} we describe an online algorithm for ordered alphabet.

\section{Catcher}\label{SectCatcher}

A \emph{string of length $n$} over an alphabet $\Sigma$ is a map $\{1,2,\ldots,n\} \mapsto \Sigma$. The length of $w$ is denoted by $|w|$. We write $w[i]$ for the $i$th letter of $w$ and $w[i..j]$ for $w[i]w[i{+}1]\ldots w[j]$. Let $w[i..i{-}1]$ be the empty string for any~$i$. A string $u$ is a \emph{substring} of $w$ if $u=w[i..j]$ for some $i$ and $j$. The pair $(i,j)$ is not necessarily unique; we say that $i$ specifies an \emph{occurrence} of $u$ in $w$. A string can have many occurrences in another string. A substring $w[1..j]$ (resp., $w[i..n]$) is a \emph{prefix} [resp. \emph{suffix}] of $w$. A string which is both a proper prefix and a suffix of $w$ is a \emph{boundary} of $w$. A \emph{square suffix} is a suffix that is a square. For any integers $i,j$, the set $\{k\in \mathbb{Z} \colon i \le k \le j\}$ (possibly empty) is denoted by $\overline{i, j}$.

Suppose $text$ is a string, $n = |text|$. To detect squares, we use an auxiliary data structure, called \emph{catcher}. The catcher works with the string $text$. For correct work, the string $text[1..n{-}1]$ must be squarefree. The catcher contains integer variables $i$, $j$ such that $0 < i \le j$ and $j - i + 1 \le n - j$ and once a letter is appended to the right of $text$, the catcher detects square suffixes beginning inside $\overline{i,j}$, i.e., if for some $k\in \overline{i,j}$, $text[k..n]$ is a square, the catcher detects this square. The segment $\overline{i,j}$ is called the \emph{trap}.

Suppose $text[k..n] = xx$ for some nonempty string $x$ and $k\in \overline{i,j}$ (see~fig.~\ref{fig:square}). Since $j - i + 1 \le n - j$, we have $x = x'x''$ for some strings $x'$ and $x''$ such that $text[j{+}1..n] = x''x'x''$.
\begin{figure}[htb]
\includegraphics[scale=0.55]{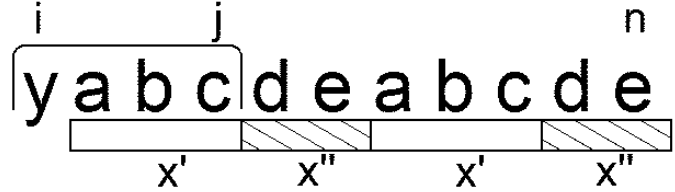}
\caption{\small The square suffix $xx$ with the leftmost position lying in $\overline{i,j}$; $j - i + 1 \le n - j$.}
\label{fig:square}
\end{figure}

\begin{lemma}
If $text[1..n{-}1]$ is squarefree, then $x''$ is the longest boundary of $text[j{+}1..n]$. \label{MaxBoundary}
\end{lemma}
\begin{proof}
Suppose $y$ is the longest boundary of $text[j{+}1..n]$ and $y \ne x''$; then $y$ is a suffix of $x$ and the occurrence of $y$ that starts at position $n{-}|x|{-}|y|{+}1$ overlaps the occurrence of $y$ that starts at position $j{+}1$. Thus, $text[1..n{-}1]$ is not squarefree. This is a contradiction. \qed
\end{proof}

Denote $t = j - i + 1$. To obtain the longest boundary of $text[j{+}1..n]$, the catcher maintains an integer array $b[j{+}1..n]$ (for convenience, we use indices $\overline{j{+}1,n}$) such that for any $k\in\overline{j{+}1,n}$, $b[k]$ is equal to the length of the longest boundary of $text[j{+}1..k]$. Further, the catcher contains a variable $s$ such that if $2b[n] + t \ge n - j$, $s$ equals the length of the longest suffix of $text[i..j]$ that is a suffix of $text[1..n{-}b[n]]$ and otherwise, $s$ equals zero. Thus by Lemma~\ref{MaxBoundary}, the catcher detects a square iff $2b[n] + s \ge n - j$.

Let us describe how to compute $b$ and $s$. There is a well-known algorithm that efficiently calculates $b[n]$ (see \cite{Stringology}). Once $b[n]$ is found, we process $s$. If $b[n] = b[n{-}1]+1$, then $s$ remains unchanged (see the definition); otherwise we put $s = 0$. Next, if $2b[n] + t \ge n - j$ and $s = 0$, we compute $s$ by a naive algorithm. The following pseudo-code summarizes the description (for convenience, we define $b[j] = -1$).
\begin{algorithmic}[1]
\State read a letter and append it to $text$ (thereby incrementing $n$)
\State $b[n] \gets b[n{-}1] + 1$
\While{$b[n] > 0 \mathrel{\mathbf{and}} text[j{+}b[n]] \ne text[n]$}
    \State $b[n] \gets b[j{+}b[n]{-}1] + 1$ \Comment{classical algorithm that calculates $b[n]$}
    \State $s \gets 0$ \Comment{this line is executed iff $b[n] \ne b[n{-}1] + 1$}
\EndWhile
\If{$2b[n] + (j - i + 1) \ge n - j \mathrel{\mathbf{and}} s = 0$} \label{lst:naivescond}
    \While{$j - s \ge i \mathrel{\mathbf{and}} text[n{-}b[n]{-}s] = text[j{-}s]$}
        \State $s \gets s + 1$ \Comment{compute $s$ by a naive algorithm}
    \EndWhile
\EndIf
\If{$2b[n] + s \ge n - j$}
    \State the square suffix of the length $2(n - j - b[n])$ is detected
\EndIf
\end{algorithmic}

\begin{lemma}
The catcher requires $O(n - j)$ time and space. \label{CatcherTime}
\end{lemma}
\begin{proof}
The algorithm that fills $b$ takes $O(n - j)$ time (see~\cite{Stringology}).

Suppose there exists a positive integer $n'$ such that our algorithm computed a nonzero value of $s$ when $text$ had the length $n'$. Let $n_1 < \ldots < n_k$ be the set of all such integers. For each $l\in \overline{1,k}$, denote by $s_l$ the value of $s$ that was computed when $text$ had the length $n_l$. It suffices to prove that $\sum_{l=1}^k s_l = O(n - j)$.
Recall that the string $text[1..n{-}1]$ is squarefree. Therefore by Lemma~\ref{MaxBoundary}, the condition in line~\ref{lst:naivescond} implies $s_l < t$ for all $l\in \overline{1,k{-}1}$. Denote $a_l = n_l - b[n_l]$ for $l\in \overline{1,k}$. It is straightforward that $\{a_l\}_{l=1}^k$ is an increasing sequence. Let $k'$ be the maximal integer such that $k'\in \overline{1,k}$ and $a_{k'} \le j + t$. Let us first prove that $\sum_{l=1}^{k'} s_l = O(t)$.

It follows from the definition of catcher that for any $l\in \overline{1,k}$, $j + t \le n_l$. Therefore for each $l \in \overline{1,k'}$, $text[a_l{+}1..j{+}t]$ is a boundary of $text[j{+}1..j{+}t]$. Let us show that $s_l < a_l - a_{l-1}$ for all $l\in \overline{2,k'}$. Suppose $s_l \ge a_l - a_{l-1}$ for some $l\in \overline{2,k'}$ (see fig.~\ref{fig:catcher1}). Denote $r = a_l - a_{l-1}$. Then $text[a_{l-1}{+}1..a_l] = text[j{-}r{+}1..j]$ by definition of $s_l$. But $text[a_{l-1}{+}1..a_l]$ is a prefix of $text[j{+}1..j{+}t]$ because $text[a_{l-1}{+}1..j{+}t]$ is a boundary of $text[j{+}1..j{+}t]$. So, we obtain a square $text[j{-}r{+}1..j{+}r]$ and this is a contradiction. Thus, $\sum_{l=1}^{k'} s_l < t + \sum_{l=2}^{k'}(a_l - a_{l-1}) = O(t)$.
\begin{figure}[htb]
\includegraphics[scale=0.55]{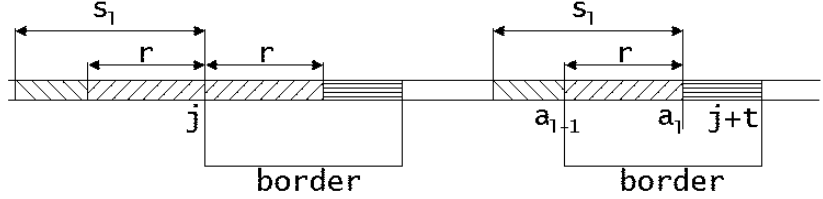}
\caption{\small If for some $l\in\overline{2,k'}$, $s_l \ge a_l - a_{l-1}$, then we have a square of the length $2r$.}
\label{fig:catcher1}
\end{figure}

To estimate the sum $\sum_{l=k'+1}^k s_l$, we first prove the following statement.
\begin{equation}
\begin{array}{l}
\text{If for some $l \in \overline{k'{+}2,k}$, $s_l \ge a_l - a_{l-1}$, then for at most one} \\
\text{$l' \in \overline{k'{+}2,l{-}1}$, the same inequality $s_{l'} \ge a_l - a_{l-1}$ holds.} \label{eqn:catcherstate}
\end{array}
\end{equation}
Suppose, to the contrary, for some $l\in \overline{k'{+}2,k}$, $s_l \ge a_l - a_{l-1}$ and there are $l_1, l_2 \in \overline{k'{+}2,l{-}1}$ such that $l_1 < l_2$ and $s_{l_1}, s_{l_2} \ge a_l - a_{l-1}$ (see fig.~\ref{fig:catcher2}). Denote $r = a_l - a_{l-1}$. By definition of $s_l$, we have $text[a_{l-1}{+}1..a_l] = text[j{-}r{+}1..j]$. Now it is easy to see that $n_{l-1} < a_l$. Indeed, if $n_{l-1} \ge a_l$, then $text[a_{l-1}{+}1..a_l] = text[j{+}1..j{+}r]$ because $text[a_{l-1}{+}1..n_l]$ is a boundary of $text[j{+}1..n_l]$; but this implies that $text[j{-}r{+}1..j{+}r]$ is a square. Thus $b[n_{l-1}] < r$. Since, by definition of $s_{l_1}$ and $s_{l_2}$, $text[a_{l_1}{-}r{+}1..a_{l_1}] = text[a_{l_2}{-}r{+}1..a_{l_2}]$ and $text[1..n{-}1]$ is squarefree, we have $a_{l_2} - a_{l_1} > r$. Recall that $a_{l_1}, a_{l_2} > j + t$ by definition of $k'$. Finally, we obtain $n_{l-1} - j \ge t + (a_{l_2} - a_{l_1}) + b[n_{l-1}] > t + r + b[n_{l-1}] > t + 2b[n_{l-1}]$. But the condition in line~\ref{lst:naivescond} suggests that $n_{l-1} - j \le t + 2b[n_{l-1}]$. This is a contradiction.
\begin{figure}[htb]
\includegraphics[scale=0.55]{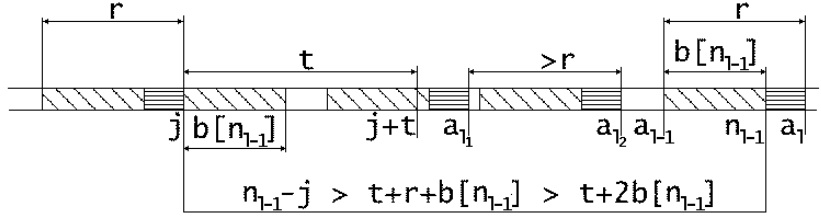}
\caption{\small Since $a_{l_2} - a_{l_1} > r$, $r > b[n_{l-1}]$, and $a_{l_1} > j + t$, we have $n_{l-1} - j > t + 2b[n_{l-1}]$.}
\label{fig:catcher2}
\end{figure}

Now we can estimate the sum $\sum_{l=k'+1}^k s_l$. If for each $l\in \overline{k'{+}2,k}$, $s_l < a_l - a_{l-1}$, then $\sum_{l=k'+1}^k s_l < t + \sum_{l=k'+2}^k (a_l - a_{l-1}) = O(n - j)$. Let $l_1 < l_2 < \ldots < l_m$ be the set of all $l'\in \overline{k'{+}2,k}$ such that $s_{l'} \ge a_{l'} - a_{l'-1}$. Denote $z = \sum_{p=1}^m s_{l_p}$. It follows from~\eqref{eqn:catcherstate} that for at most one $p \in \overline{1,m{-}1}$, $s_{l_p} \ge a_{l_m} - a_{l_m-1}$. So, $z < (m-1)(a_{l_m} - a_{l_m-1}) + 2t$. In the same way we obtain that for at most one $p \in \overline{1,m{-}2}$, $s_{l_p} \ge a_{l_{m-1}} - a_{l_{m-1}-1}$. So, $z < (m - 2)(a_{l_{m-1}} - a_{l_{m-1}-1}) + 2(a_{l_m} - a_{l_m-1}) + 2t$. Further, for at most one $p \in \overline{1,m{-}3}$, $s_{l_p} \ge a_{l_{m-2}} - a_{l_{m-2}-1}$. So, $z < (m - 3)(a_{l_{m-2}} - a_{l_{m-2}-1}) + 2(a_{l_{m-1}} - a_{l_{m-1}-1}) + 2(a_{l_m} - a_{l_m-1}) + 2t$. This process leads to the inequality $z < 2\sum_{p=1}^m (a_{l_p} - a_{l_p-1}) + 2t = O(n - j)$. Finally, we have $\sum_{l=1}^k s_l = \sum_{l=1}^{k'} s_l + \sum_{l=k'+1}^k s_l = O(t) + O(n - j) + z = O(n - j)$. \qed
\end{proof}

\section{Unordered Alphabet}\label{SectUnord}

\begin{theorem}
For unordered alphabet, there exists an online algorithm that detects squares in $O(n\log n)$ time and linear space. \label{UnorderedSquares}
\end{theorem}
\begin{proof}
Our algorithm maintains $O(\log n)$ catchers and traps of these catchers cover the string $text[1..n{-}1]$. Let $k \in \overline{0,\lfloor\log n\rfloor{-}1}$ and $p$ be the maximal integer such that $(p{+}1) 2^k \le n$. We have one or two traps of the length $2^k$: the first trap is equal to $\overline{(p{-}1)2^k{+}1,p2^k}$ and if $p$ is even, we have another trap that is equal to $\overline{(p{-}2)2^k{+}1,(p{-}1)2^k}$ (see fig.~\ref{fig:systraps}). If $n$ has became a multiple of $2^k$ after extension of $text$, we add a new trap of the length $2^k$ and destroy two previous traps of the length $2^k$ if the new $p$ is odd and these traps exist. In the following pseudo-code we use the three-operand $\mathbf{for}$ loop like in the C language.
\begin{algorithmic}[1]
\State read a letter and append it to $text$ (thereby incrementing $n$)
\For{$(k \gets 0;\; 2^k \le n/2 \mathrel{\mathbf{and}} (n \bmod 2^k) = 0;\;k \gets k + 1)$}
    \State $p \gets \lfloor n / 2^k\rfloor - 1$
    \State create a catcher with the trap $\overline{(p{-}1)2^k{+}1,p2^k}$
    \If{$(p \bmod 2) \ne 0 \mathrel{\mathbf{and}} p > 1$}
        \State remove two previous catchers with the traps of the length $2^k$
    \EndIf
\EndFor
\end{algorithmic}
\begin{figure}[htb]
\includegraphics[scale=0.55]{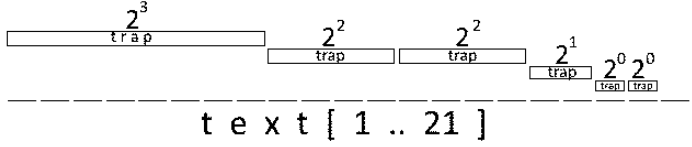}
\caption{\small The system of traps covers the string $text[1..n{-}1]$, where $n = 21$.}
\label{fig:systraps}
\end{figure}

To prove that the described system of traps covers the string $text[1..n{-}1]$, it suffices to note that if some iteration of the loop removes two catchers with the traps of the length $2^k$, then the next iteration creates a catcher with the trap of the length $2^{k+1}$ on their place.

It is easy to see that if for some $k \in \overline{0,\lfloor\log n\rfloor{-}1}$, the proposed algorithm maintains a trap $\overline{i,j}$ of length $2^k$, then $n - j < 3\cdot 2^k$. Hence it follows from Lemma~\ref{CatcherTime} that all catchers of the algorithm use $O(2\sum_{k=0}^{\lfloor\log n\rfloor{-}1} 3\cdot 2^k) = O(n)$ space. Since traps of the same length do not intersect, to maintain traps of the length $2^k$, the algorithm takes, by Lemma~\ref{CatcherTime}, $O(3\cdot 2^k \cdot \frac{n}{2^k}) = O(n)$ time. Thus, the algorithm requires $O(n\log n)$ overall time. \qed
\end{proof}

\section{Ordered Alphabet}\label{SectOrd}

In the case of ordered alphabet the following lemma narrows the area of square suffix search.

\begin{lemma}
Let $text$ be a string of length $n$. Denote by $t$ the length of the longest suffix of $text$ that occurs at least twice in $text$. If $text[1..n{-}1]$ is squarefree and for some positive $k$, $text[n{-}k{+}1..n]$ is a square, then $t < k \le 2t$. \label{SquareLocation}
\end{lemma}
\begin{proof}
Suppose $k \le t$. Since $text[n{-}t{+}1..n]$ has two occurrences in $text$, the square $text[n{-}k{+}1..n]$ occurs twice and the string $text[1..n{-}1]$ is not squarefree. This is a contradiction.

Suppose $k > 2t$. Note that $k$ is even. Then the suffix $text[n{-}k/2{+}1..n]$ has at least two occurrences in $text$. This contradicts to the definition of $t$. \qed
\end{proof}

For each integer $i \ge 0$, denote by $t_i$ the length of the longest suffix of $text[1..i]$ that has at least two occurrences in $text[1..i]$. We say that there is an online access to the sequence $\{t_i\}$ if any algorithm that reads the string $text$ sequentially from left to right can read $t_i$ immediately after reading $text[i]$.

The following lemma describes an online algorithm for square detection based on an online access to $\{t_i\}$. Note that the alphabet is not necessarily ordered.

\begin{lemma}
If there is an online access to the sequence $\{t_i\}$, then there exists an algorithm that online detects squares in linear time and space. \label{OrderedLemma}
\end{lemma}
\begin{proof}
Our algorithm online reads the string $text$ (as above, $n$ denotes the number of letters read) and maintains an integer variable $s$ such that $\frac{3}{4}s \le t_n \le s$ (initially $s = 0$). To detect a square, we use three catchers; the traps of these catchers are denoted by $\overline{i_1,j_1}$, $\overline{i_2,j_2}$, and $\overline{i_3,j_3}$. The traps satisfy the following conditions (any of these traps can be empty; see fig.~\ref{fig:threetraps}):
\begin{equation}
i_1 \le n - 2t_n + 1, j_1 + 1 = i_2, j_2 + 1 = i_3, j_3 \ge n - \frac{3}{4}s\enspace. \label{eqn:trapscond}
\end{equation}
\begin{figure}[htb]
\includegraphics[scale=0.55]{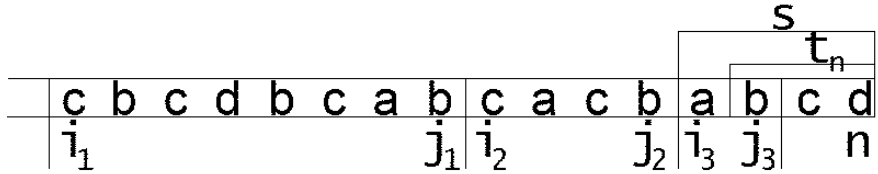}
\caption{\small The traps $\overline{i_1,j_1}$, $\overline{i_2,j_2}$, and $\overline{i_3,j_3}$ just after update; $\frac{3}4s \le t_n \le s$.}
\label{fig:threetraps}
\end{figure}
Thus, the traps cover the block $\overline{n{-}2t_n{+}1,n{-}t_n}$ and therefore, by Lemma~\ref{SquareLocation}, the algorithm detects squares. Consider the following pseudo-code that maintains the catchers and the variable $s$ (``update'' command replaces the corresponding catcher with a new one):
\begin{algorithmic}[1]
\State read a letter and append it to $text$ (thereby we increment $n$ and read $t_n$)
\State $s \gets s + 1$
\If{$\frac{3}{4}s > t_n$}
    \State $s \gets t_n$    \label{lst:updateall}
    \State $i_1 \gets +\infty$, $j_3 \gets -\infty$ \Comment{these assignments cause updates of all catchers}
\EndIf
\If{$i_1 > \max\{n - 2t_n + 1, 1\}$} \label{lst:recalc12}
    \State update the 1st catcher $i_1 \gets \max\{n - 4s + 1, 1\}$, $j_1 \gets \max\{n - 2s, 1\}$
    \State update the 2nd catcher $i_2 \gets j_1 + 1$, $j_2 \gets n - s$
\EndIf
\If{$j_3 < n - \frac{3}{4}s$} \label{lst:recalc3}
    \State update the 3rd catcher $i_3 \gets n - s + 1$, $j_3 \gets n - \lceil s/2\rceil$ \label{lst:catcher3}
\EndIf
\end{algorithmic}

Clearly, the proposed algorithm preserves~\eqref{eqn:trapscond} and thus by Lemma~\ref{SquareLocation}, works correctly. Further, it follows from Lemma~\ref{CatcherTime} that the algorithm uses linear space. To end the proof, it suffices to estimate the working time.

Suppose the algorithm processed a string of length $n$. Let for any $n'\in \overline{1,n}$, the term \emph{$n'$th step} refers to the set of instructions performed by the algorithm when it read and processed the letter $text[n']$. For any $n'\in \overline{0,n}$, denote by $s_{n'}$ the value of $s$ that was calculated on $n'$th step. Let $n_1 < n_2 < \ldots < n_k$ be the set of all steps on that the algorithm performed the line~\ref{lst:updateall}. At first we estimate the time required by steps $n_1, n_2, \ldots, n_k$; then we estimate the time required by other steps: the updates of the first and second catchers and finally the updates of the third catcher.

Consider the steps $n_1, n_2, \ldots, n_k$. It follows from Lemma~\ref{CatcherTime} that for any $i\in \overline{1,k}$, the updates of the first, second, and third catchers on $n_i$th step require $O(2s_{n_i})$, $O(s_{n_i})$, and $O(s_{n_i}/2)$ time respectively. Since $s_{n_i} < \frac{3}{4}(s_{n_i-1} + 1)$, we have $\frac{1}{3}s_{n_i} < s_{n_i-1} - s_{n_i} + 1$. Therefore, the steps $n_1, n_2, \ldots, n_k$ require $O(s_{n_1} + s_{n_2} + \ldots + s_{n_k}) = O((s_{n_1}{-}s_{n_1-1}) + (s_{n_2}{-}s_{n_2-1}) + \ldots + (s_{n_k}{-}s_{n_k-1}) + k) = O(n)$ time.

Let $i\in \overline{1,k{-}1}$. Consider the steps $n_i{+}1, n_i{+}2, \ldots, n_{i+1}{-}1$. Let $p_1 < p_2 < \ldots < p_l$ be the set of all steps $p \in \overline{n_i{+}1,n_{i+1}}$ on that the algorithm updated the first and second catchers. Clearly $p_l = n_{i+1}$. For any $j\in \overline{1,l}$, the recalculation of the first and second catchers on $p_j$th step takes $O(2s_{p_j} + s_{p_j}) = O(s_{p_j})$ time. The condition in line~\ref{lst:recalc12} implies that for any $j\in \overline{2,l}$, $p_j - p_{j-1} \ge 2s_{p_{j-1}}$. Hence the recalculation is performed in $O(s_{p_1} + s_{p_2} + \ldots + s_{p_{l-1}}) = O((p_2{-}p_1) + (p_3{-}p_2) + \ldots + (p_l{-}p_{l-1})) = O(n_{i+1} - n_i)$ time.

Let $q_1 < q_2 < \ldots < q_m$ be the set of all steps $q \in \overline{n_i{+}1, n_{i+1}}$ on that the algorithm updated the third catcher. Clearly $q_m = n_{i+1}$. For any $j\in \overline{1,m}$, this recalculation takes $O(s_{q_j})$ time on $q_j$th step. It follows from the lines~\ref{lst:recalc3}--\ref{lst:catcher3} that for any $j\in \overline{2,m}$, $q_j - q_{j-1} \ge s_{q_{j-1}}$. Then the recalculation takes $O(s_{q_1} + s_{q_2} + \ldots + s_{q_{m-1}}) = O((q_2{-}q_1) + (q_3{-}q_2) + \ldots + (q_m{-}q_{m-1})) = O(n_{i+1} - n_i)$ time.

Thus, all recalculations of the catchers are performed in $O((n_2{-}n_1) + (n_3{-}n_2) + \ldots + (n_{k}{-}n_{k-1})) = O(n)$ time and this result ends the proof. \qed
\end{proof}

\begin{theorem}
For ordered alphabet, there exists an algorithm that online detects squares in $O(n\log\sigma)$ time and linear space, where $\sigma$ is the number of different letters in the input string.
\end{theorem}
\begin{proof}
To compute the sequence $\{t_i\}$, we can use the standard Ukkonen's online algorithm \cite{Ukkonen}, which works in $O(n\log\sigma)$ time and linear space, or some more space efficient and practical algorithm (see \cite{OkanoharaSadakane} for example). Thus, the theorem follows from Lemma~\ref{OrderedLemma}. \qed
\end{proof}

\begin{corollary*}
For constant alphabet, there exists an algorithm that online detects squares in linear time and space.
\end{corollary*}

\section{Conclusion}\label{SectConclusion}

Some important problems still remain open. To date, there is no nontrivial lower bound for the problem of square detection in the case of ordered alphabet. It follows from \cite{KolpakovKucherov} that such lower bound immediately implies the same lower bound for the problem of Lempel-Ziv factorization (the later is a widely used tool in stringology and data compression).

It is also interesting to construct an efficient online algorithm for square detection that allows a ``rollback'' operation, i.e., the operation that cuts off a suffix of arbitrary length from the read string. One such algorithm is presented in~\cite{Shur}.

\bibliographystyle{splncs}

\end{document}